%% file: main.tex
\documentclass[
reprint,
nofootinbib,
superscriptaddress,
amsmath,amssymb,
aps,
pra,
floatfix,
]{revtex4-2}

\makeatletter
\def\label#1{%
  \@bsphack
  \begingroup
  \UseHookWithArguments{label}{1}{#1}%
  \protected@write\@auxout{}%
    {\string\newlabel{#1}{{\@currentlabel}{\thepage}%
      {\@currentlabelname}{\@currentHref}{\@kernel@reserved@label@data}}}%
  \endgroup
  \@esphack
}
\let\ltx@label\label
\makeatother
\usepackage{bbm}
\usepackage{amsmath,amssymb,amsthm,amscd,mathtools}
\usepackage{enumitem}
\usepackage{xcolor}
\usepackage[
    colorlinks=false,
    linkbordercolor={1 0 0},
    citebordercolor={0 1 0},
    urlbordercolor={0 1 1}
]{hyperref}

\usepackage{aliascnt}
\usepackage[nameinlink,noabbrev]{cleveref}

\theoremstyle{plain}

\newtheorem{theorem}{Theorem}[section]
\crefname{theorem}{theorem}{theorems}
\Crefname{theorem}{Theorem}{Theorems}

\theoremstyle{definition}

\newaliascnt{lemma}{theorem}
\newtheorem{lemma}[lemma]{Lemma}
\aliascntresetthe{lemma}
\crefname{lemma}{lemma}{lemmas}
\Crefname{lemma}{Lemma}{Lemmas}
\newcommand{\1}{\mathbf{1}}

\newaliascnt{proposition}{theorem}
\newtheorem{proposition}[proposition]{Proposition}
\aliascntresetthe{proposition}
\crefname{proposition}{proposition}{propositions}
\Crefname{proposition}{Proposition}{Propositions}

\newaliascnt{corollary}{theorem}
\newtheorem{corollary}[corollary]{Corollary}
\aliascntresetthe{corollary}
\crefname{corollary}{corollary}{corollaries}
\Crefname{corollary}{Corollary}{Corollaries}

\theoremstyle{definition}

\newaliascnt{definition}{theorem}

\aliascntresetthe{definition}
\crefname{definition}{definition}{definitions}
\Crefname{definition}{Definition}{Definitions}

\newaliascnt{remark}{theorem}

\aliascntresetthe{remark}
\crefname{remark}{remark}{remarks}
\Crefname{remark}{Remark}{Remarks}

\newcommand{\F}{\mathbb F}

\newcommand{\CNOT}{\mathrm{CNOT}}
\newcommand{\ellCNOT}{\ell_{\mathrm{CNOT}}}
\newcommand{\ellclean}{\ell_{\mathrm{clean}}}
\newcommand{\ellcleanr}[1]{\ell_{\mathrm{clean},#1}}
\newcommand{\ellborrow}{\ell_{\mathrm{borrow}}}
\newcommand{\ellborrowr}[1]{\ell_{\mathrm{borrow},#1}}
\newcommand{\CNOTDistance}{\mathsf{CNOT\text{-}Distance}}

\DeclareMathOperator{\GL}{GL}
\DeclareMathOperator{\Cay}{Cay}
\DeclareMathOperator{\xor}{xor}
\DeclareMathOperator{\val}{val}

\begin{document}

\title{CNOT-Distance is NP-complete under all-to-all connectivity}

\author{Antonio Acuaviva}
\author{Arturo Acuaviva}
\author{Pablo Acuaviva}

\noaffiliation

\hypersetup{
    pdftitle={CNOT-Distance is NP-complete under all-to-all connectivity},
    pdfauthor={Antonio Acuaviva, Arturo Acuaviva, Pablo Acuaviva},
    pdfsubject={Mathematics of computation and quantum circuit synthesis},
    pdfkeywords={CNOT synthesis, linear reversible circuits, NP-completeness, Vertex Cover, Cayley graph distance}
}

\date{\today}

\begin{abstract}

Given $A\in\GL(N,2)$ and an integer $K$, we ask whether $A$ can be implemented by at most $K$ CNOT gates on fixed labelled wires with all-to-all connectivity. We prove that this problem is NP-complete. From a finite simple graph $G=(V,E)$, we construct an upper-unitriangular matrix
\[
A_G\in\GL(2|V|+|E|+1,2)
\]
satisfying
\[
\ell_{\CNOT}(A_G)=2|V|+2|E|+\tau(G),
\]
where $\tau(G)$ is the minimum vertex-cover size. Each target matrix has $O(N)$ nonzero entries and row Hamming weight at most four. The lower bound unfolds an arbitrary CNOT circuit into an XOR directed acyclic graph and applies projection--contraction operations, permitting cancellation and unrestricted reuse of intermediate parities. For this family, the optimum is unchanged by any finite number of clean or borrowed ancillary wires, required to be restored. A polynomial-time decoder also yields hardness of approximation within every fixed additive constant and, via an L-reduction from Minimum Vertex Cover on cubic graphs, APX-hardness of the associated CNOT-circuit optimisation problem.
\end{abstract}

\maketitle

\input{1-Introduction}
\input{2-Preliminaries-and-definitions}
\input{3-The-reduction}
\input{4-Complexity-consequences}
\input{5-conclusion}

\section*{Statement on AI use}

Large language models, in particular OpenAI’s ChatGPT 5.6 Pro, were used in the development of this work. The authors independently verified all mathematical claims, proofs, citations, and final text.
 
\clearpage
\appendix
\input{6_Appendix}

\bibliographystyle{ieeetr}
\bibliography{bibliography}

\end{document}

%% file: 1-Introduction.tex
\section{Introduction}\label{sec:introduction}

Controlled-NOT (CNOT) gates are a central primitive in quantum-circuit synthesis. Together with arbitrary one-qubit gates, they form a universal gate library \cite{BarencoEtAl1995}. Circuits composed exclusively of CNOT gates implement precisely the invertible linear transformations over $\F_2$. Such linear-reversible transformations arise both as stand-alone synthesis targets in reversible and quantum computation and as explicit components of standard synthesis procedures for stabilizer circuits and CNOT--phase circuits \cite{AaronsonGottesman2004,AmyAzimzadehMosca2019}. CNOT-only optimisation has consequently been studied as a stand-alone synthesis problem, while CNOT and Clifford resynthesis methods have also been used as optimisation primitives within broader quantum-compilation workflows \cite{BravyiShaydulinHuMaslov2021,MurphyKissinger2023}.

There is an important distinction between constructing a short implementation and finding a shortest one. Gaussian elimination synthesises every $A\in\GL(N,2)$ using $O(N^2)$ CNOT gates. Patel, Markov, and Hayes improved the worst-case bound to $O(N^2/\log N)$, matching the information-theoretic lower bound up to a multiplicative constant \cite{PatelMarkovHayes2008}. This determines the optimal worst-case order of growth, but it does not determine the \textit{minimum} gate count required. 

Instance-wise optimisation has therefore been approached through exact distances, lower bounds, and scalable methods. Bataille computed exact small-width distances and studied special families \cite{Bataille2022}; Christensen et al. subsequently obtained the complete distance distribution through seven wires \cite{ChristensenEtAl2025}, while Bu, Fan, and Joo developed polynomial-time lower bounds~\cite{BuFanJoo2025}. Recent results of J{\o}rgensen and of Gong and Yu construct explicit families whose CNOT complexity is at least $4N-o(N)$~\cite{Jorgensen2026,GongYu2026}. Other work uses algebraic heuristics, reinforcement learning, model-based planning, and parallel exact synthesis \cite{DeBrugiereEtAl2021,RomanelloEtAl2025,CossioEtAl2026,LiEtAl2026}. These results concern algorithmic or extremal properties of the metric, and while they provide several indications of computational difficulty, they do not address the complexity of evaluating the all-to-all, fixed-label distance of a supplied matrix.

Against this algorithmic background, the available hardness results concern closely related but distinct synthesis models. Amy, Azimzadeh, and Mosca proved NP-completeness for fixed-target parity-network synthesis and for parity networks with linearly encoded inputs, while leaving the arbitrary-target, unencoded, ancilla-free parity-network problem open~\cite{AmyAzimzadehMosca2019}. The distinction is that a parity network prescribes parities that must occur during a computation, rather than a final reversible transformation; recent graphic parity-network results remain in that model~\cite{CaoEtAl2025}.

Hardness is also known when either the gate library or the architecture is changed. Van de Wetering and Amy proved, under polynomial-time Turing reductions, NP-hardness of CNOT-count minimisation when candidate circuits may use the full Clifford+$T$ gate set ~\cite{VanDeWeteringAmy2024}. Jiang et al. proved constant-factor inapproximability for global CNOT size and depth under a supplied directed topology with clean ancillas, as well as for local window-size minimisation with fixed surrounding circuitry ~\cite{JiangEtAl2020}. Closest to our setting, Kang proved NP-hardness of ancilla-free CNOT-only synthesis under a supplied directed topology; on the instances in Kang’s reduction, the constrained optimum is exactly $2|E|+2\tau(G)$~\cite{Kang2023}. Because that reduction uses the supplied topology essentially, it does not itself establish hardness under all-to-all connectivity.

These results leave a basic question unresolved: \textit{does minimum-size, fixed-label CNOT synthesis remain hard when every control--target pair is available?} The all-to-all model isolates this architecture-independent baseline. It removes routing and gate-availability constraints, so any remaining hardness arises from minimising the linear-reversible implementation itself rather than from the geometry of a particular device.

We work with $N$ prescribed labelled wires and permit a CNOT between every ordered pair of distinct wires. The target must be realised exactly, without a free input or output permutation, and no cancellation-free, monotonicity, or gate-order condition is imposed. The base model has no ancillary wires; clean and borrowed ancillary variants are considered separately. Formally, a CNOT with control $j$ and target $i$ is the coordinate transvection $T_{ij}=I_N+E_{ij}$, and $\ellCNOT(A)$ is the corresponding word length of $A\in\GL(N,2)$.

The problem $\CNOTDistance$ takes as input an explicitly represented square binary matrix $A\in\mathbb{F}_2^{N\times N}$ and an integer $K\ge 0$, and asks whether there exists a CNOT circuit with at most $K$ gates whose induced linear transformation is $A$. Our main result is the following.

\begin{theorem}\label{thm:main}\label{thm:np-complete}
The problem $\CNOTDistance$ is NP-complete. NP-hardness holds even for instances $(A,K)$ in which
\begin{enumerate}
  \item $A$ is upper unitriangular;
  \item every row of $A$ has Hamming weight at most four and $A$ has $O(N)$ nonzero entries;
  \item $(A-I_N)^3=0$ and $A$ has exact order four; and
  \item $K<2N$.
\end{enumerate}
\end{theorem}

\subsection{Proof strategy and organisation}

Our reduction is from \textsc{Vertex Cover}, which is NP-complete ~\cite{Karp1972}. Given a finite simple graph $G=(V,E)$, we construct in polynomial time an upper-unitriangular matrix
\[
    A_G\in\GL(2|V|+|E|+1,2)
\]
and prove the exact identity
\[
    \ellCNOT(A_G)=2|V|+2|E|+\tau(G),
\]
where $\tau(G)$ is the size of a minimum vertex cover of $G$.

For the upper bound, every vertex cover $C$ determines a CNOT circuit implementing $A_G$ with exactly
\[
    2|V|+2|E|+|C|
\]
gates. Selected vertex wires temporarily carry the distinguished contribution needed to form the incident edge outputs, after which all wires attain their prescribed output forms. Taking $C$ to be minimum gives the required upper bound.

The converse is the main difficulty. Under all-to-all connectivity, an arbitrary circuit need not resemble this construction: it may distribute the required parities among unrelated wires, reuse intermediate forms, and later cancel them. Thus there is no prescribed path along which the contribution of the distinguished input $z$ must travel, and no cancellation-free normal form may be assumed.

Let $L$ be the length of an arbitrary CNOT circuit implementing $A_G$. Unfolding the circuit produces an XOR directed acyclic graph with exactly $L$ internal nodes. We project selected inputs to zero and contract nodes whose projected values coincide with earlier available values. The prescribed vertex outputs force the deletion of $2|V|$ distinct nodes, while projecting the private edge inputs forces the deletion of a further $|E|$ nodes. These operations preserve all designated projected outputs. After zero-source normalisation, the remaining graph yields an unrestricted XOR straight-line program for
\[
    F_G=\{z+x_u+x_v:\{u,v\}\in E\}
\]
of length at most $L-2|V|-|E|$.

By Lemma~\ref{lem:bmp}, which is the binary specialisation of the $z$-expression lemma of Boyar, Matthews, and Peralta~\cite[Lemma~1]{BoyarMatthewsPeralta2013},
\[
    \xor(F_G)=|E|+\tau(G).
\]
Consequently,
\[
    L\ge
    2|V|+|E|+\xor(F_G)
    =2|V|+2|E|+\tau(G),
\]
which matches the upper bound without imposing any cancellation-free hypothesis. The Boyar--Matthews--Peralta identity supplies the non-reversible kernel of the reduction; the new ingredient is the projection--contraction argument that extracts this kernel from every unrestricted reversible implementation.

The same extraction applies with any finite number of clean ancillary wires by including their input variables in the projection. A borrowed implementation becomes a clean implementation when its borrowed inputs are fixed to zero. Tracking the contractions also gives the polynomial-time decoder used for the additive-inapproximability and APX-hardness consequences.

Section~\ref{sec:preliminaries} defines CNOT distance, the ancillary models, and the required XOR-circuit tools. Section~\ref{sec:reduction} constructs $A_G$ and proves the exact
distance formula and ancillary robustness. Section~\ref{sec:complexity} derives the decision, exact-computation, and approximation consequences. Appendix~\ref{app:bmp-proof} gives a self-contained binary proof of the Boyar--Matthews--Peralta kernel.

%% file: 2-Preliminaries-and-definitions.tex
\section{Notation and preliminary results}\label{sec:preliminaries}

We write $[N]=\{1,\ldots,N\}$ and denote by $I_N$ the $N\times N$ identity matrix. For $i,j\in[N]$, we denote by $E_{ij}$ the $N\times N$ matrix unit whose $(i,j)$-entry is $1$ and whose remaining entries are $0$. We identify the values on the $N$ wires with a column vector $x=(x_1,\ldots,x_N)^{\mathsf T}\in\F_2^N$. For distinct $i,j\in[N]$, the operation
\[
  x_i\leftarrow x_i+x_j
\]
has control $j$, target $i$, and matrix $T_{ij}=I_N+E_{ij}$. Equivalently, left multiplication by $T_{ij}$ adds row $j$ to row $i$.

\subsection{CNOT circuits and transvection distance}

For distinct $i,j\in[N]$, let
\begin{equation*}
    T_{ij}=I_N+E_{ij}.
\end{equation*}
For a column vector of wire values $x=(x_1,\ldots,x_N)^{\mathsf T}$, we have
\begin{equation*}
    T_{ij}x=(x_1,\ldots,x_{i-1},x_i+x_j,x_{i+1},\ldots,x_N)^{\mathsf T}.
\end{equation*}
Thus $T_{ij}$ encodes the CNOT operation
\begin{equation*}
    x_i\leftarrow x_i+x_j,
\end{equation*}
whose control is wire $j$ and whose target is wire $i$. Equivalently, left multiplication by $T_{ij}$ adds row $j$ to row $i$.

The set
\begin{equation*}
    \Sigma_N=\{T_{ij}:i,j\in[N],\ i\ne j\}
\end{equation*}
generates $\GL(N,2)$, and every generator is an involution. The minimum CNOT count of $A\in\GL(N,2)$ is therefore
\begin{equation}\label{eq:cnot-length}
    \ell_{\CNOT}(A)=\min\bigl\{L:A=T_{i_Lj_L}\ldots T_{i_1j_1},\ T_{i_tj_t}\in\Sigma_N\bigr\}.
\end{equation}
Equivalently, this is the distance from $I_N$ to $A$ in the Cayley graph
\begin{equation*}
    \Cay(\GL(N,2),\Sigma_N).
\end{equation*}
See, for example, \cite{ChristensenEtAl2025}.

We study the decision problem
\begin{equation}\label{eq:cnot-distance-problem}
    \begin{split}
        \CNOTDistance={}&\bigcup_{N\geq1}
        \bigl\{(A,K)\in\GL(N,2)\times\mathbb Z_{\geq0}:\\
        &\ell_{\CNOT}(A)\leq K\bigr\}.
    \end{split}
\end{equation}
The matrix $A$ is given explicitly by its entries and $K$ is encoded in binary. Encodings in which the first component is not a square invertible binary matrix are regarded as no-instances.

\subsection{Clean and borrowed ancillary wires}\label{subsec:ancillas}

Fix $A\in\GL(N,2)$ and $r \in \mathbb{Z}_{\geq 0}$. A CNOT circuit on $N+r$ wires is represented by some $U\in\GL(N+r,2)$, with data coordinates listed before ancillary coordinates. It is a \emph{clean-$r$ implementation} of $A$ if
\[
  U(x,0^r)=(Ax,0^r)
  \qquad\text{for every }x\in\F_2^N,
\]
and a \emph{borrowed-$r$ implementation} if
\[
  U(x,a)=(Ax,a)
  \qquad\text{for every }x\in\F_2^N,\ a\in\F_2^r.
\]
Let $\ellcleanr{r}(A)$ and $\ellborrowr{r}(A)$ be the corresponding minimum gate counts, and set
\begin{align*}
  \ellclean(A)&=\min_{r\ge0}\ellcleanr{r}(A),\\
  \ellborrow(A)&=\min_{r\ge0}\ellborrowr{r}(A).
\end{align*}
The corresponding clean-ancilla and borrowed-ancilla decision problems ask whether $\ellclean(A)\le K$ and $\ellborrow(A)\le K$, respectively, for an explicitly given matrix $A$ and binary-encoded budget $K$.

A borrowed implementation is clean after setting $a=0$, while an ancilla-free circuit may ignore any supplied ancillary wires. Therefore, for every $r$,
\begin{equation}\label{eq:ancilla-basic-ineq}
  \ellcleanr{r}(A)\le\ellborrowr{r}(A)\le\ellCNOT(A).
\end{equation}

\subsection{XOR straight-line programs}

Let $U$ be a finite set of formal input variables. A \emph{linear form in $U$} is an expression
\begin{equation*}
    f=\sum_{u\in U}\alpha_u u,
    \qquad
    \alpha_u\in\F_2.
\end{equation*}
Thus a linear form is the XOR of a subset of the variables in $U$.

An \emph{XOR straight-line program} of length $L$ begins with the input variables in $U$ available. For each $t\in[L]$, it chooses two linear forms $a_t$ and $b_t$ that are already available and produces the new linear form
\begin{equation*}
    w_t=a_t+b_t.
\end{equation*}
The forms $a_t$ and $b_t$ may be input variables or forms produced on earlier lines. Every input variable and every previously produced form remains available after the assignment. Thus, after line $t$, the available forms are
\begin{equation*}
    U\cup\{w_1,\ldots,w_t\}.
\end{equation*}
Each assignment counts as one XOR operation.

Let $F$ be a finite set of linear forms in $U$. We say that the program computes $F$ if
\begin{equation*}
    F\subseteq U\cup\{w_1,\ldots,w_L\}.
\end{equation*}
We define the \emph{XOR complexity of $F$} by
\begin{equation*}
    \begin{split}
        \xor(F)=\min\bigl\{L:\ &\text{there is a length-$L$ XOR}\\
        &\text{straight-line program computing $F$}\bigr\}.
    \end{split}
\end{equation*}

For a simple undirected graph $G=(V,E)$, introduce one distinguished input variable $z$ and one input variable $x_v$ for each $v\in V$. Thus the set of input variables is
\begin{equation*}
    U_G=\{z\}\cup\{x_v:v\in V\}.
\end{equation*}
The variable $z$ is independent of the variables $x_v$ and is not associated with any vertex of $G$. For each edge $e=\{u,v\}\in E$, define the linear form
\begin{equation*}
    f_e=z+x_u+x_v
\end{equation*}
and set
\begin{equation}\label{eq:FG}
    F_G=\{f_e:e\in E\}
       =\{z+x_u+x_v:\{u,v\}\in E\}.
\end{equation}
Let $\tau(G)$ denote the minimum size of a vertex cover of $G$.

We will use the following binary specialisation of the $z$-expression lemma of Boyar--Matthews--Peralta \cite[Lemma~1]{BoyarMatthewsPeralta2013}.

\begin{lemma}[Boyar--Matthews--Peralta]\label{lem:bmp}
For every simple graph $G = (V, E)$,
\begin{equation*}
    \xor(F_G)= |E|+\tau(G).
\end{equation*}
More strongly, from any length-$s$ linear straight-line program computing $F_G$, one can extract in polynomial time a vertex cover of size at most $s-|E|$.
\end{lemma}

\begin{proof}
Apply \cite[Lemma~1]{BoyarMatthewsPeralta2013} over $\F_2$ with input set
\begin{equation*}
    X=\{z\}\cup\{x_v:v\in V\}
\end{equation*}
and with the repetition-free family of $z$-expressions
\begin{equation*}
    F_G=\{z+x_u+x_v:\{u,v\}\in E\}.
\end{equation*}
The family is repetition-free because $G$ is simple, and $|F_G|=|E|$. Moreover, each $f_e$ contains three distinct input variables and is therefore not itself an input variable. Hence, from any program computing $F_G$ in the sense defined above, we may designate the first line producing each $f_e$ as its output line. Since the forms in $F_G$ are distinct, these are $|F_G|$ distinct lines, and this designation changes neither the program nor its length.

A cover of $F_G$ in the terminology of the cited lemma is a set of variables of the form
\begin{equation*}
    \{x_v:v\in C\}
\end{equation*}
meeting every expression $z+x_u+x_v$; this is equivalent to $C$ being a vertex cover of $G$, and the two covers have the same cardinality. The cited lemma therefore shows that a length-$s$ program computing $F_G$ yields, in polynomial time, a vertex cover of size at most $s-|F_G|=s-|E|$. Consequently every such program satisfies
\begin{equation*}
    s\geq |E|+\tau(G).
\end{equation*}
Conversely, the same lemma applied to a minimum vertex cover gives a program of length $|E|+\tau(G)$. Hence
\begin{equation*}
    \xor(F_G)=|E|+\tau(G).
\end{equation*}
The polynomial-time extraction assertion is precisely the algorithmic part of the cited lemma. For completeness, we include an independent self-contained proof over $\F_2$ in \Cref{app:bmp-proof}.
\end{proof}

\subsection{Projection and contraction}

It is convenient to view a straight-line program as a directed acyclic XOR circuit. Its source nodes are the input variables, each internal node is the XOR of its two parents, and selected source or internal nodes are designated as outputs. The same predecessor may occupy both parent slots. For any node $g$, write $\val(g)$ for the linear form computed at $g$.

For a set $S$ of input variables, write $\pi_S$ for the coordinate projection that sends every variable in $S$ to zero and fixes all other inputs. An ordering of the nodes of a directed acyclic circuit is called a \emph{topological ordering} if every node occurs after each of its parents. Equivalently, every directed edge goes from an earlier node to a later node. Such an ordering is simply an order in which the circuit can be evaluated, and every directed acyclic circuit admits one. We shall need the following results.

\begin{lemma}[Projection-contraction]\label{lem:projection-contraction}
Let $P$ be an XOR circuit, let $\pi_S$ be a coordinate projection, and let $g$ be an internal node. Suppose there is a source or an internal node $h$ occurring earlier than $g$ in a fixed topological ordering such that
\begin{equation*}
    \pi_S(\val(g))=\pi_S(\val(h)).
\end{equation*}
After applying $\pi_S$ to the source labels, one may delete $g$, redirect every outgoing edge of $g$ to $h$, and redirect any output designation on $g$ to $h$. The resulting circuit computes all projected outputs and has one fewer internal node.
\end{lemma}

\begin{proof}
After applying $\pi_S$ to the source labels, the values computed at $g$ and $h$ are equal. Delete $g$, redirect every outgoing edge of $g$ to $h$, and redirect any output designation on $g$ to $h$. Since $h$ is a source or occurs earlier than $g$ in the chosen topological ordering, these redirections create no directed cycle.

Now examine the remaining nodes in topological order. The source values are unchanged, apart from the prescribed projection. Suppose that an internal node is reached and that every preceding node computes the same projected linear form as before. The values supplied by both of its parents are then unchanged: an incoming edge that did not originate at $g$ is unaffected, while an incoming edge formerly originating at $g$ now receives from $h$ the same projected linear form. Hence the XOR computed at the node is also unchanged. It follows by induction that every remaining node computes the same projected linear form as before. The same is therefore true of every designated output, including any output formerly designated at $g$, which is now designated at $h$. Thus the resulting circuit computes all projected outputs and has one fewer internal node.
\end{proof}

\begin{lemma}[Zero-source normalisation]\label{lem:zero-source-normalisation}
Suppose an XOR circuit has ordinary sources $U$, additional sources labelled by the zero form, and designated outputs that are nonzero forms in $U$. It can be converted, without increasing its number of internal nodes, into an XOR straight-line program over the inputs $U$ that computes the same outputs.
\end{lemma}

\begin{proof}
Fix a topological ordering
\begin{equation*}
    g_1,\ldots,g_r
\end{equation*}
of the internal nodes, so that the parents of each $g_i$ are sources or belong to
\begin{equation*}
    \{g_1,\ldots,g_{i-1}\}.
\end{equation*}
For a straight-line program, this is simply the order of its instructions. We construct the new circuit inductively along this ordering. To each source and each already processed internal node $v$, associate either the formal null symbol $\bot$, when $\val(v)=0$, or a source or previously constructed internal node $\rho(v)$ that computes $\val(v)$. For each ordinary source $u\in U$, set $\rho(u)=u$, and associate every zero-labelled source with $\bot$.

Suppose that $g_1,\ldots,g_{i-1}$ have been processed, and let $p_i$ and $q_i$ be the parents of $g_i$. By the choice of the topological ordering, the representatives of $p_i$ and $q_i$ have already been determined. If
\begin{equation*}
    \val(p_i)=\val(q_i),
\end{equation*}
associate $g_i$ with $\bot$, since
\begin{equation*}
    \val(g_i)=\val(p_i)+\val(q_i)=0
\end{equation*}
over $\F_2$. If exactly one of $\val(p_i)$ and $\val(q_i)$ is zero, associate $g_i$ with the representative of the other parent. Otherwise, $\val(p_i)$ and $\val(q_i)$ are distinct nonzero forms. In this case, append one XOR assignment whose inputs are $\rho(p_i)$ and $\rho(q_i)$, and define $\rho(g_i)$ to be its output. This output computes
\begin{equation*}
    \val(p_i)+\val(q_i)=\val(g_i).
\end{equation*}

Proceeding through $g_1,\ldots,g_r$ therefore introduces at most one new internal node for each original internal node. Since every designated output is nonzero, it has a representative when the construction terminates. These representatives compute all the required outputs using only the sources in $U$.
\end{proof}

%% file: 3-The-reduction.tex
\section{The reduction}\label{sec:reduction}

The reduction separates a graph-independent synthesis cost from the graph-dependent term $\tau(G)$, the size of a minimum vertex cover. Let $G=(V,E)$ be a finite simple graph, and write
\[
  n=|V|,\qquad m=|E|.
\]
We construct a transformation on
\[
  N=2n+m+1
\]
labelled data wires. There is a distinguished wire $z$, two wires $q_v,x_v$ for each $v\in V$, and one wire $y_e$ for each $e\in E$. The $q_v$ and $y_e$ coordinates are arbitrary inputs with prescribed outputs; they are not ancillary workspace.

Define $A_G$ by
\begin{align}
  z'&=z,\label{eq:z-output}\\
  q_v'&=q_v+z &&(v\in V),\label{eq:q-output}\\
  x_v'&=x_v+q_v &&(v\in V),\label{eq:x-output}\\
  y_{\{u,v\}}'&=y_{\{u,v\}}+z+x_u+x_v
      &&(\{u,v\}\in E).\label{eq:y-output}
\end{align}

The construction has two interacting parts. Each edge coordinate $y_{\{u,v\}}$ must acquire the parity
\[
  f_{\{u,v\}}=z+x_u+x_v;
\]
these parities form precisely the family $F_G$ from Eq.~\eqref{eq:FG}. For each vertex, the pair $(q_v,x_v)$ provides reversible bookkeeping. In the upper-bound construction, selecting $v$ allows $x_v$ to carry $z$ temporarily to incident edge outputs. In the lower bound, the prescribed vertex outputs force two contractions per vertex, while each private edge input $y_e$ forces one further contraction.

Choose fixed orderings of $V$ and $E$. Let $B\in\F_2^{m\times n}$ be the edge--vertex incidence matrix, so that $B_{e,v}=1$ precisely when $v$ is incident with $e$, and let $\1_r$ be the all-one column vector of length $r$. Writing
\[
  Y=(y_e)_{e\in E},\qquad
  X=(x_v)_{v\in V},\qquad
  Q=(q_v)_{v\in V},
\]
the defining equations become
\begin{align*}
  Y'&=Y+BX+\1_m z,\\
  X'&=X+Q,\\
  Q'&=Q+\1_n z,\\
  z'&=z.
\end{align*}
In particular, $(BX)_{\{u,v\}}=x_u+x_v$. Hence, in the coordinate order $(Y,X,Q,z)$,
\begin{equation}\label{eq:AG}
  A_G=
  \begin{pmatrix}
    I_m & B   & 0   & \1_m\\
    0   & I_n & I_n & 0\\
    0   & 0   & I_n & \1_n\\
    0   & 0   & 0   & 1
  \end{pmatrix}.
\end{equation}

\begin{proposition}[Structural restrictions]\label{prop:structure}
The matrix $A_G$ is upper unitriangular, has exactly $4m+4n+1$ nonzero entries, and every row has Hamming weight at most four. If $M=A_G-I_N$, then
\[
  M^3=0.
\]
If $n>0$, then $M^2\ne0$ and $A_G$ has exact order four.
\end{proposition}

\begin{proof}
Upper unitriangularity is immediate from Eq.~\eqref{eq:AG}. Every edge row has four nonzero entries; every $x_v$- and $q_v$-row has two; and the $z$-row has one. This gives $4m+4n+1=O(N)$ nonzero entries and maximum row weight four.

The strict upper part is
\[
  M=
  \begin{pmatrix}
    0 & B & 0 & \1_m\\
    0 & 0 & I_n & 0\\
    0 & 0 & 0 & \1_n\\
    0 & 0 & 0 & 0
  \end{pmatrix},
\]
so
\[
  M^2=
  \begin{pmatrix}
    0 & 0 & B & 0\\
    0 & 0 & 0 & \1_n\\
    0 & 0 & 0 & 0\\
    0 & 0 & 0 & 0
  \end{pmatrix}.
\]
The only potentially nonzero block of $M^3$ is $B\1_n$. Each row of the incidence matrix contains exactly two ones, hence $B\1_n=0$ over $\F_2$ and $M^3=0$. If $n>0$, the $(X,z)$ block $\1_n$ shows that $M^2\ne0$. In characteristic two,
\begin{align*}
  (I_N+M)^2&=I_N+M^2\ne I_N,\\
  (I_N+M)^4&=I_N+M^4=I_N.
\end{align*}
Thus $A_G$ has exact order four.
\end{proof}

\begin{theorem}[Exact distance formula]\label{thm:exact-distance}
For every finite simple graph $G=(V,E)$,
\[
  \ellCNOT(A_G)=2|V|+2|E|+\tau(G).
\]
\end{theorem}

The two directions have different content. For the upper bound, a vertex cover determines an explicit circuit in which selected vertex wires temporarily carry $z$. For the lower bound, we unfold an arbitrary CNOT circuit, remove the fixed reversible contribution by projection and contraction, and recover an XOR straight-line program for $F_G$.

\subsection{Upper bound}\label{subsec:upper-bound}

\begin{lemma}\label{lem:upper-bound}
If $C\subseteq V$ is a vertex cover, then
\[
  \ellCNOT(A_G)\le 2n+2m+|C|.
\]
\end{lemma}

\begin{proof}
List the vertices of $C$ in a fixed order
$v_1,\ldots,v_c$, where $c=|C|$, and initially mark every
edge as unprocessed. For each wire labelled $w$, let
$w^{\mathrm{in}}$ denote its input variable, and use $w$ for
the linear form currently carried by that wire. Thus, initially,
$w=w^{\mathrm{in}}$. None of the gates in the construction
targets $z$, so
\[
  z=z^{\mathrm{in}}
\]
throughout.

For each $r=1,\ldots,c$, first apply
\begin{equation}\label{eq:upper-cover-prep}
  q_{v_r}\leftarrow q_{v_r}+z,
  \qquad
  x_{v_r}\leftarrow x_{v_r}+z.
\end{equation}
Then, for every unprocessed edge $e=\{v_r,u\}$ incident with
$v_r$, apply
\begin{equation}\label{eq:upper-edge}
  y_e\leftarrow y_e+x_{v_r},
  \qquad
  y_e\leftarrow y_e+x_u,
\end{equation}
and mark $e$ as processed.

When $e=\{v_r,u\}$ is processed, Eq.~\eqref{eq:upper-cover-prep}
has made
\[
  x_{v_r}=x_{v_r}^{\mathrm{in}}+z^{\mathrm{in}}.
\]
The other endpoint $u$ cannot be an earlier vertex in the
ordered cover, because in that case $e$ would already have
been processed. Hence $u$ either lies outside $C$ or appears
later in the ordering, and therefore
\[
  x_u=x_u^{\mathrm{in}}.
\]
Moreover, because $e$ is still unprocessed, no preceding gate
has targeted $y_e$, so $y_e=y_e^{\mathrm{in}}$ immediately
before the two gates in Eq.~\eqref{eq:upper-edge}. Consequently,
these gates produce
\[
  y_e
  =y_e^{\mathrm{in}}
   +z^{\mathrm{in}}
   +x_{v_r}^{\mathrm{in}}
   +x_u^{\mathrm{in}},
\]
which is the required output for the edge wire. Every edge is
eventually processed because $C$ is a vertex cover.

After all edge wires have been processed, apply
\begin{equation}\label{eq:upper-cover-finish}
  x_v\leftarrow x_v+q_v
  \qquad(v\in C).
\end{equation}
For each $v\in C$, Eq.~\eqref{eq:upper-cover-prep} gives
\[
  x_v=x_v^{\mathrm{in}}+z^{\mathrm{in}},
  \qquad
  q_v=q_v^{\mathrm{in}}+z^{\mathrm{in}}.
\]
Thus Eq.~\eqref{eq:upper-cover-finish} cancels the two copies
of $z^{\mathrm{in}}$ and produces
\[
  x_v=x_v^{\mathrm{in}}+q_v^{\mathrm{in}}.
\]

For every $v\notin C$, apply, in the indicated order,
\begin{equation}\label{eq:upper-noncover}
  x_v\leftarrow x_v+q_v,
  \qquad
  q_v\leftarrow q_v+z.
\end{equation}
Before these gates, neither $x_v$ nor $q_v$ has been modified:
the edge-stage gates use $x_v$ only as a control. The two gates
therefore produce
\[
  x_v=x_v^{\mathrm{in}}+q_v^{\mathrm{in}},
  \qquad
  q_v=q_v^{\mathrm{in}}+z^{\mathrm{in}}.
\]
For $v\in C$, the required value of $q_v$ was already produced
by Eq.~\eqref{eq:upper-cover-prep}. Together with the edge
outputs established above and the unchanged wire $z$, this
verifies every output equation defining $A_G$.

The construction uses three gates for each vertex in $C$, two
gates for each vertex outside $C$, and two gates for each edge.
Its total length is therefore
\[
  3|C|+2(n-|C|)+2m
  =2n+2m+|C|.
\]
\end{proof}

Taking $C$ minimum gives
\begin{equation}\label{eq:upper-final} \ellCNOT(A_G)\le 2n+2m+\tau(G). \end{equation}

\subsection{Lower bound}\label{subsec:lower-bound}

The lower bound has two steps. Unfolding represents each CNOT gate by one XOR node, and robust extraction then removes $2n+m$ forced nodes, leaving a straight-line program for $F_G$.

\begin{lemma}[Unfolding]\label{lem:unfolding}
A length-$L$ CNOT circuit on a fixed set of wires unfolds into an XOR circuit with exactly $L$ internal nodes whose designated outputs are the final wire forms.
\end{lemma}

\begin{proof}
Use one source for each initial wire. Whenever a CNOT replaces a target form $r$ by $r+s$, create a fresh XOR node with the current target and control nodes as parents, and redirect the target wire to that node. Exactly one node is created per gate. Induction over the gate sequence shows that every wire pointer always carries its current form.
\end{proof}

The following lemma isolates the main lower-bound mechanism. For the ancilla-free case, one may read it with $\mathcal A=\varnothing$ and $\lambda_v=\mu_v=\nu_e=0$. We state the more general form because the same extraction will later apply to circuits with clean ancillary wires.

\begin{lemma}[Robust extraction]\label{lem:robust-extraction}
Let $P$ be an XOR circuit with $L$ internal nodes and sources
\[
  \{z\}\cup\{x_v,q_v:v\in V\}\cup\{y_e:e\in E\}\cup\mathcal A,
\]
where $\mathcal A$ is any finite set of additional variables. Suppose that its designated outputs include
\begin{align*}
  \widetilde a_v&=q_v+z+\lambda_v,\\
  \widetilde b_v&=x_v+q_v+\mu_v,\\
  \widetilde g_{\{u,v\}}&=y_{\{u,v\}}+z+x_u+x_v+\nu_{\{u,v\}},
\end{align*}
where $\lambda_v,\mu_v$, and $\nu_e$ are linear forms in the variables $\mathcal A$. In polynomial time, $P$ can be transformed into an XOR straight-line program for $F_G$ of length at most
\[
  L-2n-m.
\]
\end{lemma}

\begin{proof}
Discard every output designation except the displayed vertex and edge outputs. This changes neither the circuit nor its number of internal nodes, and it is harmless because the conclusion only requires a straight-line program computing $F_G$. We perform three stages, each preserving the currently projected retained outputs.

\emph{Vertex projections.}
Project every $q_v$ and every variable in $\mathcal A$ to zero. Before projection, the $2n$ displayed vertex-output forms are pairwise distinct and none is a source, because their data-variable parts are pairwise distinct sums of two variables. They therefore occupy $2n$ distinct internal nodes. After projection,
\[
  \widetilde a_v\longmapsto z,
  \qquad
  \widetilde b_v\longmapsto x_v.
\]
Contract each of those nodes to the corresponding, earlier source using Lemma~\ref{lem:projection-contraction}. This deletes exactly $2n$ internal nodes. Each edge output now has the form
\[
  y_e+f_e.
\]

\emph{Private-edge projections.}
Process the variables $y_e$ sequentially. Before $y_e$ is projected, its designated output $y_e+f_e$ still depends on it and differs from the source $y_e$, so an internal node with nonzero $y_e$ coefficient exists. In a topological ordering of the current projected-and-contracted circuit, choose the earliest such internal node. The source $y_e$ is the only earlier node with a nonzero $y_e$ coefficient; hence exactly one parent of the chosen node is that source, while the other parent computes a form $h$ independent of $y_e$. After projecting $y_e$ to zero, the chosen node and the other parent both compute $h$, so Lemma~\ref{lem:projection-contraction} deletes the chosen node. The designated edge output becomes $f_e$; the other edge-output forms are unchanged because they do not contain the private variable $y_e$. The argument remains valid even if other computations used cancelling copies of $y_e$, because projection--contraction preserves every designated projected output. Repeating this step for all $e\in E$ deletes $m$ pairwise distinct internal nodes and leaves $F_G$ among the retained designated outputs.

\emph{Normalisation.} 
Discard every output designation except those corresponding to the forms in $F_G$. This changes neither the circuit nor its number of internal nodes and is harmless because the conclusion only requires a straight-line program computing $F_G$. All sources in $\{q_v:v\in V\}\cup\{y_e:e\in E\}\cup\mathcal A$ now carry the zero form, while every retained designated output is a nonzero form in the surviving sources $U=\{z\}\cup\{x_v:v\in V\}$. Lemma~\ref{lem:zero-source-normalisation} therefore removes the zero sources without increasing the number of internal nodes and yields an XOR straight-line program over $z$ and the $x_v$ that computes $F_G$. The accounting is
\begin{align*}
  L&\longrightarrow L-2n
      &&\text{after the vertex contractions},\\
   &\longrightarrow L-2n-m
      &&\text{after the private-edge contractions}.
\end{align*}
All projections and contractions can be implemented using coefficient vectors and a maintained topological ordering, so the transformation runs in polynomial time.
\end{proof}

\begin{lemma}[Projected lower bound]\label{lem:projected-lower-bound}
If a CNOT circuit implementing $A_G$ has length $L$, then
\[
  L\ge 2n+m+\xor(F_G)=2n+2m+\tau(G).
\]
\end{lemma}

\begin{proof}
Unfold the circuit using Lemma~\ref{lem:unfolding} and apply Lemma~\ref{lem:robust-extraction} with $\mathcal A=\varnothing$. The resulting program for $F_G$ has length at most $L-2n-m$, so $L-2n-m\ge\xor(F_G)$. Lemma~\ref{lem:bmp} gives $\xor(F_G)=m+\tau(G)$.
\end{proof}

The lower bound in Lemma~\ref{lem:projected-lower-bound} and the upper bound in Eq.~\eqref{eq:upper-final} coincide, proving Theorem~\ref{thm:exact-distance}.

The exact identity proves the core reduction. We now record two strengthenings needed for the complexity consequences: constructive decoding from arbitrary implementations and robustness to ancillary wires.

\begin{proposition}[Constructive decoder]\label{prop:decoder}
Given $G$ and a length-$L$ CNOT circuit implementing $A_G$, one can construct in polynomial time a vertex cover $C$ of $G$ satisfying
\[
  |C|\le L-2n-2m.
\]
\end{proposition}

\begin{proof}
Unfold the circuit and invoke the polynomial-time transformation in Lemma~\ref{lem:robust-extraction} with $\mathcal A=\varnothing$. It produces a straight-line program for $F_G$ of length $s\le L-2n-m$. The algorithmic part of Lemma~\ref{lem:bmp} returns a vertex cover of size at most $s-m\le L-2n-2m$.
\end{proof}

\begin{proposition}[Ancillary wires do not help]\label{prop:ancillas}
For every graph $G$ and every $r\ge0$,
\begin{align}\label{eq:ancilla-equality}
  \ellcleanr{r}(A_G)
  &=\ellborrowr{r}(A_G)=\ellCNOT(A_G)\notag\\
  &=2n+2m+\tau(G).
\end{align}
The same value is obtained after minimising over an arbitrary finite number of clean or borrowed ancillary wires. Moreover, the decoder of Proposition~\ref{prop:decoder} applies in polynomial time to circuits in either ancillary model.
\end{proposition}

\begin{proof}
The ancilla-free upper-bound circuit ignores the ancillary wires, so Eq.~\eqref{eq:ancilla-basic-ineq} supplies the upper bounds. It remains to prove the clean lower bound, since every borrowed implementation is clean.

Let a clean-$r$ circuit have length $L$ and full linear transformation $U$. With data coordinates first, the clean condition implies the block form
\[
  U=
  \begin{pmatrix}
    A_G & R\\
    0   & S
  \end{pmatrix}.
\]
Consequently, on arbitrary ancillary inputs $\alpha$, the designated data outputs are the corresponding rows of $A_G$ plus linear forms in $\alpha$. Unfolding the circuit therefore gives exactly the situation of Lemma~\ref{lem:robust-extraction}, with the ancillary variables as $\mathcal A$. That lemma and Lemma~\ref{lem:bmp} imply
\[
  L\ge2n+m+\xor(F_G)=2n+2m+\tau(G).
\]
Composing the same extraction with the decoder in Lemma~\ref{lem:bmp} proves the polynomial-time claim. Thus the clean lower bound holds, and therefore so does the borrowed lower bound.
\end{proof}

%% file: 4-Complexity-consequences.tex
\section{Complexity consequences}\label{sec:complexity}

The exact-distance identity transfers several complexity properties of
\textsc{Vertex Cover} to CNOT synthesis. We first prove the decision result and
its ancillary variants, and then derive hardness for exact computation,
fixed-additive approximation, and multiplicative approximation.

\subsection{Decision and exact computation}

We begin with the decision problem stated in Theorem~\ref{thm:main}. The exact-distance formula gives the reduction directly. The remaining points are to preserve the promised structural restrictions and to show that polynomial-size certificates exist even when the gate budget is encoded in binary.

\begin{proof}[Proof of Theorem~\ref{thm:main}]
For NP-hardness, reduce from \textsc{Vertex Cover} on finite simple graphs,
which is NP-complete~\cite{Karp1972}. Let $(G,k)$ be an instance and
write $n_0=|V(G)|$.

If $k>n_0$, the instance is trivially positive. In this case, output the
fixed positive CNOT instance obtained from the one-vertex edgeless graph
$G_{\mathrm{yes}}$, with
\[
  N_{\mathrm{yes}}=3,
  \qquad
  K_{\mathrm{yes}}=2.
\]
Theorem~\ref{thm:exact-distance} gives
\[
  \ellCNOT(A_{G_{\mathrm{yes}}})=2.
\]
Proposition~\ref{prop:structure} shows that this instance satisfies
restrictions (i)--(iii), and
\[
  K_{\mathrm{yes}}=2<6=2N_{\mathrm{yes}}.
\]

We may therefore assume that $0\le k\le n_0$. If $V(G)=\varnothing$,
add one isolated vertex. This preserves $\tau(G)$ and the answer while
ensuring that the resulting graph has at least one vertex. Write
\[
  n=|V(G)|,
  \qquad
  m=|E(G)|
\]
for the resulting graph, construct $A_G$, and set
\[
  K=2n+2m+k.
\]
The matrix and the budget are computable in polynomial time.
Theorem~\ref{thm:exact-distance} gives
\[
  \tau(G)\le k
  \quad\Longleftrightarrow\quad
  \ellCNOT(A_G)\le K.
\]
Proposition~\ref{prop:structure} supplies restrictions (i)--(iii).
Because $k\le n$ and $n\ge1$,
\[
  K\le3n+2m<4n+2m+2=2N,
\]
which proves restriction (iv).

It remains to establish membership in NP. A positive instance is
necessarily invertible, and invertibility can be checked in polynomial
time. Gauss--Jordan elimination reduces every
$A\in\GL(N,2)$ to $I_N$ using at most $N(N-1)$ row additions and at
most $N$ row swaps. Over $\F_2$, a row swap is implemented by three row
additions. Thus every invertible matrix has a CNOT implementation of
length at most
\[
  N(N-1)+3N=N^2+2N,
\]
which is $O(N^2)$. Since every row addition is an involution, reversing
the elimination sequence implements $A$.

Consequently, if $(A,K)$ is a yes-instance and $L$ is the length of a
shortest implementation, then
\[
  L\le\min\{K,N^2+2N\}.
\]
A polynomial-size certificate therefore exists even when $K$ is
encoded in binary. The certificate is a list of ordered
control--target pairs. The verifier checks that its length is at most
$K$, applies the corresponding row additions to $I_N$, and compares
the resulting matrix with $A$.
\end{proof}

The same reduction establishes hardness in the ancillary models because Proposition~\ref{prop:ancillas} shows that ancillary wires do not reduce the optimum on the constructed instances. For completeness, we also verify that allowing an unspecified number of ancillary wires does not obstruct membership in NP.

\begin{corollary}[Ancillary decision problems]
\label{cor:ancilla-np}
The clean- and borrowed-ancilla decision problems are NP-complete under
restrictions (i)--(iv) of Theorem~\ref{thm:main}, even when an
implementation may choose any finite number of ancillary wires.
\end{corollary}

\begin{proof}
Proposition~\ref{prop:ancillas}, together with the reduction above,
gives NP-hardness.

For membership in NP, let $L$ be the length of a shortest valid
implementation of a yes-instance. Since an ancilla-free elimination
circuit is available in either model,
\[
  L\le\min\{K,N^2+2N\}.
\]
A length-$L$ circuit mentions at most $2L$ ancillary-wire labels.
Deleting every unmentioned ancillary wire and relabelling the remaining
ones consecutively gives an equivalent certificate with
\[
  r\le2L\le2(N^2+2N).
\]
The certificate consists of $r$ and the gate list.

From the gate list, the verifier computes the full transformation $U$
on the $N+r$ wires and checks that the list has length at most $K$.
With the data coordinates listed first, the clean condition is
equivalent to
\[
  U=
  \begin{pmatrix}
    A&R\\
    0&S
  \end{pmatrix}
\]
for some blocks $R$ and $S$. The borrowed condition is equivalent to
\[
  U=\operatorname{diag}(A,I_r).
\]
These conditions, the wire indices, and the gate count can all be
checked in time polynomial in the certificate size.
\end{proof}

The decision reduction also yields hardness of computing the optimum itself. Here only the matrix restrictions (i)--(iii) are relevant, because restriction (iv) concerns the decision budget $K$. 

\begin{corollary}[Exact computation]\label{cor:exact-hard}
Computing $\ellCNOT(A)$ exactly is NP-hard, even for matrices satisfying
restrictions (i)--(iii) of Theorem~\ref{thm:main}. The same holds for
the clean- and borrowed-ancilla optima. Equivalently, computing the
distance from the identity in
\[
  \operatorname{Cay}(\GL(N,2),\Sigma_N)
\]
is NP-hard when the target matrix is given explicitly by its entries.
\end{corollary}

\begin{proof}
Exact minimum vertex-cover size remains NP-hard when the input graph is
required to have at least one vertex, since an isolated vertex may be
added without changing the optimum. For every such graph,
Theorem~\ref{thm:exact-distance} gives
\[
  \tau(G)
  =\ellCNOT(A_G)-2|V(G)|-2|E(G)|.
\]
Thus exact CNOT-distance computation would compute $\tau(G)$.
Equation~\eqref{eq:ancilla-equality} gives the same conclusion for
$\ellclean(A)$ and $\ellborrow(A)$.
\end{proof}

\subsection{Approximation consequences}

The exact identity and the constructive decoder yield two different forms of approximation hardness. First, disjoint-union amplification rules out every fixed additive error. Second, the decoder converts the excess length of an arbitrary CNOT implementation into the excess size of a vertex cover, providing the solution map required for an L-reduction.

\begin{corollary}[Fixed additive approximation]
\label{cor:additive-hard}
For every fixed integer $c\ge0$, it is NP-hard, given an explicitly
represented matrix $A\in\GL(N,2)$, to output a polynomially encoded
rational number $d(A)$ satisfying
\[
  |d(A)-\ellCNOT(A)|\le c,
\]
even for matrices satisfying restrictions (i)--(iii) of
Theorem~\ref{thm:main}. The same conclusion holds with $\ellCNOT$
replaced by $\ellclean$ or $\ellborrow$.

In particular, it is NP-hard to output a feasible implementation whose
length is at most the corresponding optimum plus $c$.
\end{corollary}

\begin{proof}
Assume that such a polynomial-time algorithm exists, and set
\[
  t=2c+1.
\]
Given a graph $G$ with at least one vertex, let $H$ be the disjoint
union of $t$ copies of $G$. Then
\begin{align*}
  |V(H)|&=t|V(G)|,\\
  |E(H)|&=t|E(G)|,\\
  \tau(H)&=t\tau(G).
\end{align*}
Set
\[
  D_H=2|V(H)|+2|E(H)|.
\]
Theorem~\ref{thm:exact-distance} gives
\[
  \ellCNOT(A_H)=D_H+t\tau(G).
\]
Applying the assumed algorithm to $A_H$ therefore yields
\[
  \left|
    \frac{d(A_H)-D_H}{t}-\tau(G)
  \right|
  \le\frac{c}{2c+1}<\frac12.
\]
Rounding to the nearest integer recovers $\tau(G)$ exactly in
polynomial time.

The matrices $A_H$ satisfy restrictions (i)--(iii), and
Proposition~\ref{prop:ancillas} identifies the same optimum in the
clean- and borrowed-ancilla models. Finally, if an algorithm outputs a
feasible circuit of length $L$ satisfying
\[
  L\le\operatorname{OPT}+c,
\]
then feasibility gives $L\ge\operatorname{OPT}$. Taking $d(A)=L$
would therefore satisfy the preceding additive guarantee.
\end{proof}

For the multiplicative approximation result, we make the polynomial
solution encoding explicit. Set
\[
  B(N)=N(N-1)+3N=N^2+2N.
\]
Define \textsc{Minimum-CNOT} to be the minimisation problem whose
instances are explicitly represented nonidentity matrices
$A\in\GL(N,2)$ and whose feasible solutions are CNOT gate lists of
length at most $B(N)$ implementing $A$. The objective is the length of
the gate list. The identity is excluded only to avoid a zero optimum
in the multiplicative approximation problem.

The Gauss--Jordan construction used above shows that
\[
  \ellCNOT(A)\le B(N)
\]
for every invertible $A$. Hence the length bound does not remove an
optimal solution or change the optimum. Instance recognition,
feasibility, and gate count are polynomial-time computable, and every
feasible solution has polynomial encoding length. Thus
\textsc{Minimum-CNOT} is an NPO minimisation problem.

Define the clean- and borrowed-ancilla variants analogously. A feasible
solution additionally specifies an integer
\[
  0\le r\le2B(N),
\]
uses consecutively labelled ancillary wires, has length at most $B(N)$,
and satisfies the corresponding restoration condition. These bounds
also leave the unrestricted ancillary optima unchanged. Indeed, an
unrestricted optimal ancillary implementation has length
\[
  L^*\le\ellCNOT(A)\le B(N).
\]
After unused ancillary wires are deleted and the remaining ones are
relabeled, it uses at most
\[
  r\le2L^*\le2B(N)
\]
ancillary wires. Feasibility is verified using the same block conditions
as in the proof of Corollary~\ref{cor:ancilla-np}. Hence these are also
well-defined NPO minimisation problems with the unrestricted optima
$\ellclean(A)$ and $\ellborrow(A)$.

We can now state the multiplicative approximation consequence. The length bounds above serve only to give polynomially bounded solution encodings; they do not alter any of the three optima.

\begin{corollary}[APX-hardness]\label{cor:apx-hard}
The ancilla-free, clean-ancilla, and borrowed-ancilla minimisation
problems defined above are APX-hard under L-reductions, even for
matrices satisfying restrictions (i)--(iii) of
Theorem~\ref{thm:main}. Consequently, none admits a polynomial-time
approximation scheme unless $\mathrm{P}=\mathrm{NP}$.
\end{corollary}

\begin{proof}
 The minimum vertex cover is APX-complete even on cubic graphs, in which
every vertex has degree three~\cite{AlimontiKann2000}. Let
$G=(V,E)$ be a nonempty cubic instance and write
\[
  n=|V|,
  \qquad
  m=|E|,
  \qquad
  D_G=2n+2m.
\]
Let $C^*$ be a minimum vertex cover. Since $G$ is cubic,
\[
  2m=3n.
\]
Every edge has at least one endpoint in $C^*$, so
\[
  m
  \le\sum_{v\in C^*}\deg(v)
  =3\tau(G).
\]
It follows that $n\le2\tau(G)$. Consequently,
Theorem~\ref{thm:exact-distance} gives
\begin{align*}
  \operatorname{OPT}(A_G)
  &=D_G+\tau(G)
  =2n+2m+\tau(G)\\
  &=5n+\tau(G)
  \le11\tau(G).
\end{align*}
This is the first L-reduction condition, with $\alpha=11$.

Now let $\mathcal Q$ be any feasible target circuit for $A_G$, and let
$L$ be its length. Proposition~\ref{prop:decoder} constructs in
polynomial time a vertex cover $C$ satisfying
\[
  |C|\le L-D_G.
\]
Because every vertex cover has size at least $\tau(G)$,
\begin{align*}
  0\le |C|-\tau(G)
  &\le L-D_G-\tau(G) =L-\operatorname{OPT}(A_G).
\end{align*}
Thus the instance map
\[
  G\longmapsto A_G
\]
and the solution map that applies the constructive decoder to
$\mathcal Q$ form an L-reduction with constants
\[
  \alpha=11,
  \qquad
  \beta=1,
\]
in the sense of Papadimitriou and
Yannakakis~\cite{PapadimitriouYannakakis1991}.
Proposition~\ref{prop:structure} supplies restrictions (i)--(iii) for
the target matrices.

For clean and borrowed ancillary circuits,
Proposition~\ref{prop:ancillas} gives both the same target optimum and
a polynomial-time decoder satisfying the same bound for every feasible
implementation. The identical inequalities therefore establish
L-reductions with $\alpha=11$ and $\beta=1$ in both ancillary models.
\end{proof}

%% file: 5-conclusion.tex
\section{Conclusion}\label{sec:conclusion}

We have proved that \textsc{CNOT-Distance} is NP-complete in the fixed-label, all-to-all model. Thus exact CNOT-count minimisation remains computationally hard even after routing constraints and restricted gate availability are removed. Hardness already holds for sparse upper-unitriangular targets whose rows have Hamming weight at most four, for which
\[
    (A-I_N)^3=0,
\]
and which have exact order four. No cancellation-free, monotonicity, or prescribed gate-order assumption is imposed. 

For every finite simple graph \(G\), our construction satisfies the exact identity
\[
    \ell_{\CNOT}(A_G)
      = 2|V(G)|+2|E(G)|+\tau(G).
\]

The projection--contraction argument extracts the graph-dependent XOR kernel from an arbitrary reversible implementation, despite cancellation and unrestricted reuse of intermediate parities. The accompanying polynomial-time decoder transfers more than decision hardness: exact computation of the CNOT distance is NP-hard, approximation within every fixed additive constant is NP-hard, and the associated optimisation problem is APX-hard. For the constructed family \(A_G\), the optimum is also unchanged by any finite number of clean or borrowed ancillary wires required to be restored.

The result also clarifies what does and does not follow for topology-constrained synthesis. If a connectivity graph is supplied as part of the input, NP-hardness follows immediately by choosing the complete bidirected graph. More generally, consider any architecture family with the following uniformity property: given $t$, one can construct in time polynomial in $t$ a member of the family of size polynomial in $t$, together with an explicitly identified complete bidirected subgraph on $t$ vertices. Given a graph $G$, take $t=2|V(G)|+|E(G)|+1$, place the transformation $A_G$ on these $t$ vertices, and extend it by the identity on all remaining wires. The upper-bound circuit from the reduction uses only gates within the complete bidirected subgraph. Conversely, every architecture-respecting circuit implementing this extended transformation is, upon treating the remaining wires as borrowed ancillas, an unrestricted borrowed-ancilla implementation of $A_G$. Proposition~\ref{prop:ancillas} therefore gives the same lower bound, so the optimum is unchanged and the reduction proves NP-hardness for every architecture family satisfying this uniformity property.

Because the complete bidirected graph is an admissible topology when the connectivity graph is supplied as part of the input, the theorem also implies NP-hardness for that general topology-input model. It does not, however, determine the complexity on any particular fixed sparse architecture: restricting the allowed control--target pairs changes the generating set and hence the distance being minimised. Determining the complexity of exact fixed-label CNOT minimisation on, for example, a bidirected path or a bidirected star is therefore a natural problem not resolved by the present reduction.

Further directions include determining whether the all-to-all problem admits a constant-factor approximation, developing exact or parameterised algorithms for narrower structural classes of matrices, and studying alternative objectives such as CNOT depth or weighted gate count. Our results concern worst-case exact sequential gate-count minimisation; they do not establish hardness for typical instances or for approximate realisation of a target transformation. Finally, the equality between the ancilla-free, clean-ancilla, and borrowed-ancilla optima is established for the reduction family \(A_G\), not for arbitrary linear-reversible transformations.

%% file: 6_Appendix.tex
\section{A proof of the binary \texorpdfstring{$z$}{z}-expression lemma}\label{app:bmp-proof}

We give a self-contained proof of \Cref{lem:bmp} specialised to $\F_2$. In addition to establishing the exact XOR complexity, the proof gives a polynomial-time procedure that extracts a vertex cover from any straight-line program computing the target family. The argument permits cancellation, repeated intermediate values, repeated computations of target forms, and unrestricted reuse of earlier values.

\begin{proof}[Proof of \Cref{lem:bmp}]
We first prove the upper bound. Let $C\subseteq V$ be a vertex cover. For each $v\in C$, compute the form
\begin{equation*}
    h_v=z+x_v.
\end{equation*}
For every edge $e=\{u,v\}\in E$, choose an endpoint belonging to $C$, say $u\in C$, and compute
\begin{equation*}
    f_e=h_u+x_v=z+x_u+x_v.
\end{equation*}
This uses $|C|+|E|$ XOR assignments. Taking $C$ to be a minimum vertex cover gives
\begin{equation*}
    \xor(F_G)\leq |E|+\tau(G).
\end{equation*}

For the converse, let $P$ be a length-$s$ XOR straight-line program computing $F_G$, and write its lines as
\begin{equation*}
    w_t=p_t+q_t
    \qquad(t\in[s]),
\end{equation*}
where $p_t$ and $q_t$ are forms already available before line $t$. Call a line \emph{special} if its value belongs to $F_G$, and call it \emph{auxiliary} otherwise. Because $G$ is simple, the $|E|$ forms in $F_G$ are pairwise distinct, and none is an input variable. Since one program line has only one value, each target form must occur as the value of a distinct special line, although a target may be computed more than once. Consequently, at least $|E|$ lines are special. If $a$ denotes the number of auxiliary lines, then
\begin{equation*}
    a\leq s-|E|.
\end{equation*}

We first record a property of the parents of a special line. For a linear form $r$, let $[y]r\in\F_2$ denote the coefficient of the input variable $y$ in $r$, and define
\begin{equation*}
    \sigma(r)=\left([z]r,\sum_{v\in V}[x_v]r\right)\in\F_2^2,
\end{equation*}
where the sum is taken in $\F_2$. The map $\sigma$ is linear, and
\begin{equation*}
    \sigma(z)=(1,0),
    \qquad
    \sigma(x_v)=(0,1),
    \qquad
    \sigma(f_e)=(1,0).
\end{equation*}

We first prove a parent property: every special line has a parent produced by an earlier auxiliary line. Let $w_t\in F_G$ be the value of a special line. Since
\[
    [z]p_t+[z]q_t=[z]w_t=1,
\]
the two parents have different $z$-coefficients. Relabel them as $p$ and $q$ so that $[z]p=0$ and $[z]q=1$.

If $p$ is not an input variable, then, because it is already available, it is the value of an earlier line. That line is auxiliary, since $[z]p=0$, whereas every special-line value has
$z$-coefficient one. Otherwise $p=x_v$ for some $v\in V$. In this remaining case, linearity of $\sigma$ gives
\[
    \sigma(q)=\sigma(w_t)+\sigma(x_v)=(1,1).
\]
This is neither an input signature nor the signature of a special-line value. Hence $q$ is the value of an earlier auxiliary line. This proves the parent property.

Let $\mathcal L$ be the $(|V|+1)$-dimensional vector space of linear forms in the input variables $U_G$. We process the program lines in their given order. After processing the first $t$ lines, we construct a quotient map
\begin{equation*}
    \rho_t:\mathcal L\longrightarrow\mathcal Q_t
\end{equation*}
and a basis $\mathcal B_t$ of $\mathcal Q_t$. If $a_t$ denotes the number of auxiliary lines among the first $t$ lines, we maintain the following properties:
\begin{enumerate}
    \item The image of every input variable and every value $w_i$ with $i\leq t$ is either zero or one element of $\mathcal B_t$.
    \item Every auxiliary-line value $w_i$ with $i\leq t$ has image zero.
    \item Every element of $\mathcal B_t$ is the image of at least one input variable.
    \item We have
    \begin{equation*}
        \dim\mathcal Q_t\geq |V|+1-a_t.
    \end{equation*}
\end{enumerate}

For $t=0$, take $\mathcal Q_0=\mathcal L$, let $\rho_0$ be the identity map, and take the input variables in $U_G$ as the basis $\mathcal B_0$. All four properties then hold.

Suppose that the properties hold after line $t-1$. First assume that line $t$ is auxiliary. By the first property, each of $\rho_{t-1}(p_t)$ and $\rho_{t-1}(q_t)$ is either zero or an element of $\mathcal B_{t-1}$. Define
\begin{equation*}
    \mathcal Q_t
    =
    \mathcal Q_{t-1}\big/\big\langle\rho_{t-1}(w_t)\big\rangle
\end{equation*}
and let $\rho_t$ be the composite of $\rho_{t-1}$ with the quotient map. In other words, we impose the relation that the value of the present auxiliary line is zero.

There are only three possible nontrivial effects. If one parent image is zero and the other is a basis element, that basis element becomes zero. If the parent images are two distinct basis elements, those two basis elements become equal. If the parent images are equal, their sum is already zero and no new relation is imposed. We may therefore choose a basis $\mathcal B_t$ such that every earlier value still has image zero or one basis element, every basis element is still represented by an input variable, and
\begin{equation*}
    \dim\mathcal Q_t\geq\dim\mathcal Q_{t-1}-1.
\end{equation*}
The new auxiliary value has image zero by construction, so all four properties are preserved.

Now suppose that line $t$ is special. In this case, make no further quotient and set
\begin{equation*}
    \mathcal Q_t=\mathcal Q_{t-1},
    \qquad
    \rho_t=\rho_{t-1},
    \qquad
    \mathcal B_t=\mathcal B_{t-1}.
\end{equation*}
By the parent property proved above, one parent of this special line was produced by an earlier auxiliary line. Its image is therefore zero by the second property. The image of the special-line value consequently equals the image of its other parent and is thus either zero or one element of $\mathcal B_t$. The four properties are again preserved.

After all $s$ lines have been processed, write
\begin{equation*}
    \mathcal Q=\mathcal Q_s,
    \qquad
    \rho=\rho_s,
    \qquad
    \mathcal B=\mathcal B_s,
    \qquad
    k=\dim\mathcal Q=|\mathcal B|.
\end{equation*}
Only the $a$ auxiliary lines can decrease the dimension, and each decreases it by at most one. Hence
\begin{equation*}
    k\geq |V|+1-a.
\end{equation*}

For the images of the input variables, write
\begin{equation*}
    c_z=\rho(z),
    \qquad
    c_v=\rho(x_v)
    \qquad(v\in V).
\end{equation*}
Each of these images is either zero or one element of the basis $\mathcal B$. For every edge $e=\{u,v\}\in E$, the target $f_e$ occurs as the value of a special line, so
\begin{equation*}
    c_z+c_u+c_v
    =
    \rho(f_e)
    \in\{0\}\cup\mathcal B.
\end{equation*}
It follows that at least two of $c_z,c_u,c_v$ are equal. Indeed, if they were pairwise distinct and one were zero, their sum would be the sum of two distinct basis elements. If they were pairwise distinct and all were nonzero, their sum would be the sum of three distinct basis elements. By linear independence of $\mathcal B$, neither sum belongs to $\{0\}\cup\mathcal B$.

We now use these images as colours and construct an independent set. For every $\beta\in\mathcal B$, let
\begin{equation*}
    U_\beta=\{y\in U_G:\rho(y)=\beta\}.
\end{equation*}
Every set $U_\beta$ is nonempty because every element of $\mathcal B$ is represented by an input variable. If $c_z=0$, choose one vertex $v_\beta\in V$ satisfying $x_{v_\beta}\in U_\beta$ for every $\beta\in\mathcal B$. If $c_z\neq0$, make such a choice for every $\beta\in\mathcal B\setminus\{c_z\}$. The required vertex exists in each case because $z$ is the only input variable not indexed by a vertex, and its image is either zero or $c_z$. Let $I$ be the set of chosen vertices.

The set $I$ is independent. Indeed, distinct vertices $u,v\in I$ have distinct nonzero colours $c_u$ and $c_v$, and neither of these colours equals $c_z$. Thus $c_z,c_u,c_v$ are pairwise distinct. If $\{u,v\}$ were an edge, this would contradict the conclusion that at least two of the three colours associated with every edge must be equal.

If $c_z=0$, then $I$ contains one vertex for each of the $k$ basis elements. If $c_z\neq0$, it contains one vertex for each of the $k-1$ basis elements different from $c_z$. In either case,
\begin{equation*}
    |V|-|I|\leq |V|+1-k.
\end{equation*}
The complement
\begin{equation*}
    C=V\setminus I
\end{equation*}
is therefore a vertex cover satisfying
\begin{equation*}
    |C|
    =|V|-|I|
    \leq |V|+1-k
    \leq a
    \leq s-|E|.
\end{equation*}
Consequently every length-$s$ program computing $F_G$ satisfies
\begin{equation*}
    s\geq |E|+\tau(G).
\end{equation*}
Together with the upper bound, this proves
\begin{equation*}
    \xor(F_G)=|E|+\tau(G).
\end{equation*}

It remains to verify the algorithmic assertion. Represent each linear form by its coefficient vector in $\F_2^{|V|+1}$. A single pass through the program computes the vector associated with every line and determines whether that line is special. The quotient construction can then be implemented using a disjoint-set data structure on the input variables together with a distinguished zero class. At an auxiliary line, two different nonzero parent classes are merged, a nonzero class paired with the zero class is merged into the zero class, and equal parent classes require no operation; the auxiliary line itself is then recorded as zero. At a special line, one parent is already in the zero class, so the line is recorded as having the class of its other parent. At termination, choose the vertex representatives defining $I$ and return $V\setminus I$. Every step requires polynomial time, and the preceding argument proves that the returned set is a vertex cover of size at most $s-|E|$.
\end{proof}

%% file: bibliography.bib
@article{BoyarMatthewsPeralta2013,
  author = {Boyar, Joan and Matthews, Philip and Peralta, Ren{\'e}},
  title = {Logic minimization techniques with applications to cryptology},
  journal = {Journal of Cryptology},
  volume = {26},
  number = {2},
  pages = {280--312},
  year = {2013},
  doi = {10.1007/s00145-012-9124-7},
  url = {https://doi.org/10.1007/s00145-012-9124-7},
  note = {\href{https://doi.org/10.1007/s00145-012-9124-7}{doi:10.1007/s00145-012-9124-7}}
}

@incollection{ChristensenEtAl2025,
  author = {Christensen, Jens Emil and J{\o}rgensen, S{\o}ren Fuglede and Pavlogiannis, Andreas and van de Pol, Jaco},
  title = {On exact sizes of minimal {CNOT} circuits},
  booktitle = {Reversible Computation},
  editor = {Gl{\"u}ck, Robert and Kaarsgaard, Robin},
  series = {Lecture Notes in Computer Science},
  volume = {15716},
  publisher = {Springer},
  address = {Cham},
  pages = {71--88},
  year = {2025},
  doi = {10.1007/978-3-031-97063-4_6},
  url = {https://doi.org/10.1007/978-3-031-97063-4_6},
  note = {\href{https://doi.org/10.1007/978-3-031-97063-4_6}{doi:10.1007/978-3-031-97063-4\_6}}
}

@techreport{Kang2023,
  author = {Kang, Yifan},
  title = {{CNOT}-optimal circuit synthesis},
  institution = {Research Science Institute, Massachusetts Institute of Technology},
  year = {2023},
  url = {https://math.mit.edu/documents/rsi/2023KangY.pdf},
  note = {\href{https://math.mit.edu/documents/rsi/2023KangY.pdf}{available online}}
}

@incollection{Karp1972,
  author = {Karp, Richard M.},
  title = {Reducibility among combinatorial problems},
  editor = {Miller, Raymond E. and Thatcher, James W. and Bohlinger, Jean D.},
  booktitle = {Complexity of Computer Computations},
  publisher = {Plenum Press},
  address = {New York},
  pages = {85--103},
  year = {1972},
  doi = {10.1007/978-1-4684-2001-2_9},
  url = {https://doi.org/10.1007/978-1-4684-2001-2_9},
  note = {\href{https://doi.org/10.1007/978-1-4684-2001-2_9}{doi:10.1007/978-1-4684-2001-2\_9}}
}

@article{PatelMarkovHayes2008,
  author = {Patel, Ketan N. and Markov, Igor L. and Hayes, John P.},
  title = {Optimal synthesis of linear reversible circuits},
  journal = {Quantum Information and Computation},
  volume = {8},
  number = {3--4},
  pages = {282--294},
  year = {2008},
  doi = {10.26421/QIC8.3-4-4},
  url = {https://doi.org/10.26421/QIC8.3-4-4},
  note = {\href{https://doi.org/10.26421/QIC8.3-4-4}{doi:10.26421/QIC8.3-4-4}}
}

@inproceedings{RomanelloEtAl2025,
  author = {Romanello, Riccardo and Lizzio Bosco, Daniele and Cossio, Jacopo and Sutulovic, Dusan and Serra, Giuseppe and Piazza, Carla and Burelli, Paolo},
  title = {{CNOT} minimal circuit synthesis: A reinforcement learning approach},
  booktitle = {2025 IEEE International Conference on Quantum Artificial Intelligence (QAI)},
  publisher = {IEEE},
  pages = {253--260},
  year = {2025},
  doi = {10.1109/QAI63978.2025.00047},
  url = {https://doi.org/10.1109/QAI63978.2025.00047},
  note = {\href{https://doi.org/10.1109/QAI63978.2025.00047}{doi:10.1109/QAI63978.2025.00047}}
}

@misc{VanDeWeteringAmy2024,
  author = {van de Wetering, John and Amy, Matthew},
  title = {Optimising quantum circuits is generally hard},
  howpublished = {arXiv:2310.05958v3 [quant-ph]},
  year = {2024},
  doi = {10.48550/arXiv.2310.05958},
  url = {https://doi.org/10.48550/arXiv.2310.05958},
  note = {\href{https://doi.org/10.48550/arXiv.2310.05958}{doi:10.48550/arXiv.2310.05958}}
}

@article{BarencoEtAl1995,
  author  = {Barenco, Adriano and Bennett, Charles H. and Cleve, Richard
             and DiVincenzo, David P. and Margolus, Norman and Shor, Peter
             and Sleator, Tycho and Smolin, John A. and Weinfurter, Harald},
  title   = {Elementary Gates for Quantum Computation},
  journal = {Physical Review A},
  volume  = {52},
  number  = {5},
  pages   = {3457--3467},
  month   = nov,
  year    = {1995},
  doi     = {10.1103/PhysRevA.52.3457},
  url     = {https://doi.org/10.1103/PhysRevA.52.3457}
}

@article{AaronsonGottesman2004,
  author  = {Aaronson, Scott and Gottesman, Daniel},
  title   = {Improved Simulation of Stabilizer Circuits},
  journal = {Physical Review A},
  volume  = {70},
  number  = {5},
  pages   = {052328},
  month   = nov,
  year    = {2004},
  doi     = {10.1103/PhysRevA.70.052328},
  url     = {https://doi.org/10.1103/PhysRevA.70.052328}
}

@article{AmyAzimzadehMosca2019,
  author  = {Amy, Matthew and Azimzadeh, Parsiad and Mosca, Michele},
  title   = {On the Controlled-{NOT} Complexity of
             Controlled-{NOT}--Phase Circuits},
  journal = {Quantum Science and Technology},
  volume  = {4},
  number  = {1},
  pages   = {015002},
  year    = {2018},
  doi     = {10.1088/2058-9565/aad8ca},
  url     = {https://doi.org/10.1088/2058-9565/aad8ca},
  eprint  = {1712.01859},
  archivePrefix = {arXiv},
  primaryClass  = {quant-ph}
}

@article{BravyiShaydulinHuMaslov2021,
  author  = {Bravyi, Sergey and Shaydulin, Ruslan and Hu, Shaohan
             and Maslov, Dmitri},
  title   = {Clifford Circuit Optimization with Templates and
             Symbolic Pauli Gates},
  journal = {Quantum},
  volume  = {5},
  pages   = {580},
  month   = nov,
  year    = {2021},
  doi     = {10.22331/q-2021-11-16-580},
  url     = {https://doi.org/10.22331/q-2021-11-16-580},
  eprint  = {2105.02291},
  archivePrefix = {arXiv},
  primaryClass  = {quant-ph}
}

@inproceedings{MurphyKissinger2023,
  author    = {Murphy, Ewan and Kissinger, Aleks},
  title     = {Global Synthesis of {CNOT} Circuits with Holes},
  booktitle = {Proceedings of the 20th International Conference on
               Quantum Physics and Logic ({QPL} 2023)},
  editor    = {Mansfield, Shane and Valiron, Beno{\^i}t
               and Zamdzhiev, Vladimir},
  series    = {Electronic Proceedings in Theoretical Computer Science},
  volume    = {384},
  pages     = {75--88},
  publisher = {Open Publishing Association},
  year      = {2023},
  doi       = {10.4204/EPTCS.384.5},
  url       = {https://doi.org/10.4204/EPTCS.384.5},
  eprint    = {2308.16496},
  archivePrefix = {arXiv},
  primaryClass  = {quant-ph}
}

@misc{Jorgensen2026,
  author        = {J{\o}rgensen, S{\o}ren Fuglede},
  title         = {Lower Bounds for the {CNOT}-Complexity of
                   Linear Reversible Operators},
  year          = {2026},
  eprint        = {2607.22248},
  archivePrefix = {arXiv},
  primaryClass  = {quant-ph},
  doi           = {10.48550/arXiv.2607.22248},
  url           = {https://doi.org/10.48550/arXiv.2607.22248},
  note          = {arXiv:2607.22248}
}

@article{BuFanJoo2025,
  author  = {Bu, Alan and Fan, Evan and Joo, Robert},
  title   = {Minimum Synthesis Cost of {CNOT} Circuits},
  journal = {Quantum Information Processing},
  volume  = {24},
  number  = {7},
  pages   = {208},
  month   = jul,
  year    = {2025},
  doi     = {10.1007/s11128-025-04831-5},
  url     = {https://doi.org/10.1007/s11128-025-04831-5},
  eprint  = {2408.07898},
  archivePrefix = {arXiv},
  primaryClass  = {quant-ph}
}

@inproceedings{JiangEtAl2020,
  author    = {Jiang, Jiaqing
               and Sun, Xiaoming
               and Teng, Shang-Hua
               and Wu, Bujiao
               and Wu, Kewen
               and Zhang, Jialin},
  title     = {Optimal Space--Depth Trade-Off of {CNOT} Circuits
               in Quantum Logic Synthesis},
  booktitle = {Proceedings of the Fourteenth Annual
               {ACM}--{SIAM} Symposium on Discrete Algorithms ({SODA})},
  pages     = {213--229},
  publisher = {Society for Industrial and Applied Mathematics},
  year      = {2020},
  doi       = {10.1137/1.9781611975994.13},
  url       = {https://doi.org/10.1137/1.9781611975994.13},
  eprint    = {1907.05087},
  archivePrefix = {arXiv},
  primaryClass  = {quant-ph}
}

@misc{LiEtAl2026,
  author        = {Li, Chenjian
                   and Gao, Dingchao
                   and Zhou, Xiangzhen
                   and Guan, Ji
                   and Zhu, Pengcheng
                   and Chu, Zhufei},
  title         = {Parallelizable Exact Synthesis of Quantum Circuits
                   via Semi-Tensor Product},
  year          = {2026},
  eprint        = {2607.24195},
  archivePrefix = {arXiv},
  primaryClass  = {quant-ph},
  doi           = {10.48550/arXiv.2607.24195},
  url           = {https://doi.org/10.48550/arXiv.2607.24195},
  note          = {arXiv:2607.24195}
}

@misc{CossioEtAl2026,
  author        = {Cossio, Jacopo and
                   Lizzio Bosco, Daniele and
                   Romanello, Riccardo and
                   Serra, Giuseppe and
                   Piazza, Carla},
  title         = {{AlphaCNOT}: Learning {CNOT} Minimization
                   with Model-Based Planning},
  year          = {2026},
  eprint        = {2604.13812},
  archivePrefix = {arXiv},
  primaryClass  = {cs.AI},
  doi           = {10.48550/arXiv.2604.13812}
}

@article{DeBrugiereEtAl2021,
  author = {Goubault de Brugi\`ere, Timoth\'ee and Baboulin, Marc and Valiron, Beno\^it and Martiel, Simon and Allouche, Cyril},
  title = {Gaussian elimination versus greedy methods for the synthesis of linear reversible circuits},
  journal = {ACM Transactions on Quantum Computing},
  volume = {2},
  number = {3},
  pages = {11:1--11:26},
  year = {2021},
  doi = {10.1145/3474226}
}

@misc{GongYu2026,
  author = {Gong, Sherry and Yu, Andrew},
  title = {Explicit matrices over $\mathbb{Z}_2$ with {CNOT} and row complexity $4n-\mathrm{o}(n)$ and local logic gates},
  year = {2026},
  eprint = {2607.28598},
  archivePrefix = {arXiv},
  primaryClass = {quant-ph},
  url = {https://arxiv.org/abs/2607.28598}
}

@article{Bataille2022,
  author = {Bataille, Marc},
  title = {Quantum circuits of {CNOT} gates: optimization and entanglement},
  journal = {Quantum Information Processing},
  volume = {21},
  pages = {269},
  year = {2022},
  doi = {10.1007/s11128-022-03577-8},
  eprint = {2009.13247},
  archivePrefix = {arXiv},
  primaryClass = {quant-ph}
}

@misc{CaoEtAl2025,
  author = {Cao, Yixin and Lu, Yiren and Nie, Junhong and Sun, Xiaoming and Tian, Guojing},
  title = {Toward Minimum Graphic Parity Networks},
  howpublished = {arXiv preprint arXiv:2509.10070},
  year = {2025},
  eprint = {2509.10070},
  archivePrefix = {arXiv},
  primaryClass = {quant-ph},
  url = {https://arxiv.org/abs/2509.10070}
}

@article{AlimontiKann2000,
  author = {Alimonti, Paola and Kann, Viggo},
  title = {Some {APX}-completeness results for cubic graphs},
  journal = {Theoretical Computer Science},
  volume = {237},
  number = {1--2},
  pages = {123--134},
  year = {2000},
  doi = {10.1016/S0304-3975(98)00158-3}
}

@article{PapadimitriouYannakakis1991,
  author = {Papadimitriou, Christos H. and Yannakakis, Mihalis},
  title = {Optimization, approximation, and complexity classes},
  journal = {Journal of Computer and System Sciences},
  volume = {43},
  number = {3},
  pages = {425--440},
  year = {1991},
  doi = {10.1016/0022-0000(91)90023-X}
}
